\documentclass[10pt,twocolumn,romanappendices]{IEEEtran} 

\usepackage[utf8]{inputenc} 
\usepackage[intlimits]{amsmath}
\usepackage{amsfonts}%
\usepackage{amssymb}%
\usepackage[final]{graphicx}
\usepackage{cite}
\usepackage{setspace}
\usepackage{url}
\usepackage{bm}
\usepackage{color}
\usepackage{booktabs}
\usepackage{dsfont}
\usepackage[caption=false]{subfig}
\usepackage{algorithm}
\usepackage{algorithmic}

\newtheorem{theorem}{Theorem}[section]
\newtheorem{corollary}[theorem]{Corollary}

\newtheorem{proof}[theorem]{Proof}

\newcommand{\algorithmicinit}{\textbf{Initialization}}
\newcommand{\INIT}{\item[\algorithmicinit]}
\newcommand{\algorithmicinput}{\textbf{Input }}
\newcommand{\INPUT}{\item[\algorithmicinput]}

\newcommand{\algorithmiciter}{\textbf{Iterative 3D analysis }}
\newcommand{\ANALYSIS}{\item[\algorithmiciter]}

\begin{document}

\title{Hierarchical Coupled Geometry Analysis for Neuronal Structure and Activity Pattern Discovery}
\date{} 

\author{Gal~Mishne, Ronen~Talmon, Ron~Meir, Jackie~Schiller, Uri~Dubin and Ronald~R.~Coifman
\thanks{G. Mishne, R. Talmon and R. Meir are with the Department of Electrical Engineering, Technion - Israel Institute of Technology, Haifa 32000, Israel. (e-mail: galga@tx.technion.ac.il; ronen@ee.technion.ac.il; meir@ee.technion.ac.il). J. Schiller and U. Dubin are with the Department of Physiology, Technion Medical School, Bat-Galim, Haifa 31096, Israel. (email: jackie@tx.technion.ac.il, uri.dubin@gmail.com). R. R. Coifman is with the Department of Mathematics, Yale University, New Haven, CT 06520 USA (e-mail:ronald.coifman@math.yale.edu; coifman@math.yale.edu).}
}

\maketitle

\begin{abstract}
In the wake of recent advances in experimental methods in neuroscience, the ability to record in-vivo neuronal activity from awake animals has become feasible.
The availability of such rich and detailed physiological measurements calls for the development of advanced data analysis tools, as commonly used techniques do not suffice to capture the spatio-temporal network complexity.
In this paper, we propose a new hierarchical coupled geometry analysis, which exploits the hidden connectivity structures between neurons and the dynamic patterns at multiple time-scales.
Our approach gives rise to the joint organization of neurons and dynamic patterns in data-driven hierarchical data structures.
These structures provide local to global data representations, from local partitioning of the data in flexible trees through a new multiscale metric to a global manifold embedding.
The application of our techniques to in-vivo neuronal recordings demonstrate the capability of extracting neuronal activity patterns and identifying temporal trends, associated with particular behavioral events and manipulations introduced in the experiments.
\end{abstract}

\section{Introduction}
\label{sec:intro}
A fundamental goal in neuroscience is to understand how information is represented, stored and modified in cortical networks. 
New experimental methods in neuroscience not only enable chronic, minimally invasive, recordings of large populations of neurons with cellular level resolution, but also allow recordings from identified neuronal subtypes~\cite{Shepherd2013}. 
The ability to acquire complex large-scale detailed behavioral and neuronal datasets calls for the development of advanced data analysis tools, as commonly used techniques do not suffice to capture the spatio-temporal network complexity. 
Such a framework should deal effectively with the challenging characteristics of neuronal and behavioral data, namely connectivity structures between neurons and dynamic patterns at multiple time-scales.

Due to natural and physical constraints, the accessible high-dimensional data often exhibit geometric structures and lie on a low-dimensional manifold. 
Manifold learning is a class of data driven methods; these methods aim to find meaningful geometry-based non-linear representations that parametrize the manifold underlying the data~\cite{Tenenbaum:2000, Roweis:2000, Belkin2003, Donoho2003,Coifman2006}. 
Only very recently have we begun to witness seeds of its applicability to real biological data, and, in particular, to neuroscience (e.g.,~\cite{Vogelstein2014,Cunningham2014}). 
Yet, most existing manifold learning methods are unable to deal with the complex data sets arising in neuroscience, since they do not account for several fundamental characteristics of the structures and patterns underlying such data. 
First, current methods are sensitive to noise and interferences.
Second, to a large extent, they do not accommodate the dynamical patterns underlying the neuronal activity. 
Third, manifold learning does not take into account co-dependencies between neuronal connectivity structures and dynamical patterns.

Previous work has addressed analysis of data exhibiting such co-dependencies.
To exemplify the generality of such a problem, consider the Netflix prize~\cite{Bennett2007}, where it is desired to provide systematic suggestions and recommendations to viewers.
A co-organization enables to both group together viewers based on their similar tastes and, at the same time, group together movies based on their similar ratings across viewers.
This clustering of viewers or of movies can be highly dependent on the particular viewer, and on the particular movie; two viewers may be similar under one metric, since they both like similar adventure movies, but at the same time, quite different since they do not like the same comedies. 
Thus, the suggestion system needs different metrics for recommending different types of movies to different viewers. 

Data arising in such settings can be viewed as a 2D matrix, where in the Netflix Prize the first dimension is the viewers (observations) and the second is the movies (variables).  
The need for matrix co-organization arises when observations are not independent and identically distributed, i.e., correlations exist among both observations and variables of the data matrix. 
Similar settings also arise in analysis of documents, psychological questionnaires, gene expression data, etc., where there is no particular reason to prefer treating one dimension as independent, while the other is treated as dependent~\cite{Cheng2000, Tang2001, Busygin2008, Chi2014}.
To address problems of this sort, Gavish and Coifman~\cite{Gavish2012} proposed a methodology for matrix organization relying on the construction of a tree-based bi-organization of the data.

The analysis of natural data poses an even greater challenge, since such data may also depend on a massive number of marginally relevant variables, including distortions and unrelated measurements, requiring metrics, which are not sensitive to such variability, and which are capable of suppressing noise or irrelevant factors. 
In particular, it is insufficient to represent neuronal activity recordings, that were acquired in repetitive trials, as a 2D matrix; representing observations (time samples in one dimension), and variables (neurons in the other dimension), does not take into account the multiple scales of the time samples, since the time exhibits both local (within trial) and global (across trials) time scales.
Therefore, the data are viewed as a 3D database whose dimensions are the neurons, the local time frames and the global trial indices.

In this paper, to accommodate the three-dimensional nature of this data, we extend~\cite{Gavish2012} to a triple-geometry analysis obtaining a nonparametric model for data tensors.
We propose a completely data-driven analysis of a given rank-3 tensor that provides a co-organization of the data, i.e., each of the dimensions is organized so that it varies smoothly along the other two dimensions.
Specifically, we focus on trial-based neuronal data, however, our approach is general and can be used to analyze other types of 3D data-sets.

In addition to the challenge of organizing the data, applying manifold learning methods requires a ``good'' metric between points, which conveys local similarity, as in the Netflix example.
Regular metrics do not perform well due to the high dimensionality and hierarchical structure of the trial-based data, as well as their inability to encompass the 3D nature of the data.
For example, using the Euclidean distance or cosine similarity between two sensors, treats the neuronal recordings as a 1D vector, and does not take into account the combined local and global nature of the trial-based experiments.
Using more sophisticated metrics such as the Mahalanobis or PCA-based distances proposed in~\cite{Singer2008, Talmon2013, Talmon2014,Haddad2014,Mishne2014b}, requires a notion of locality, which is non-trivial in the given application, as the data do not necessarily follow a regular Euclidean 3D grid.
Therefore, we also address the problem of defining a new multiscale metric, that incorporates the coupling between the dimensions based on the hierarchical structure of the data in each dimension.

Broadly, our focus is on finding a good description of the data; our analysis enables us to build intrinsic data-driven multiscale hierarchical structures.
In particular, our analysis builds three types of data structures, conveying a local to global representation, from hierarchical clustering of the data to a multiscale metric to a global embedding.
These three structures are constructed in an iterative refinement procedure for each dimension, based on the other two dimensions.
Thus, we exploit the coupling between the dimensions.

At the micro-scale, we learn a multiscale organization of the data, so that each point is organized in a bottom-up hierarchical structure, using a partition tree.
Thus, each point belongs to a set of nested folders, where each folder defines a coarser and coarser sense of locality/neighborhood.

At the intermediate scale, the hierarchical organization of the data is then used to define a novel 2D multiscale metric between points.
This metric enables to organize each dimension based on a coarse-to-fine decomposition of the other two dimensions.
Thus, the metric respects the hierarchy and compares points not only based on the raw measurements, but also based on their values across scales.
It is based on a mathematical foundation, stemming from the approximation of the earth-mover's distance (EMD) proposed by Leeb~\cite{Leeb2015}.
We show that this metric is equivalent to the $l_1$ distance between points after applying a data-adaptive filter-bank.
We extend the tree-based metric to a bi-tree multiscale metric and corresponding 2D filter bank.

At the macro scale, the local organization of the data and the multiscale metric enable the calculation of a global manifold representation of the data, via the construction of an affinity kernel and its eigendecomposition~\cite{Coifman2006}.
This representation can then be used to provided a single smooth organization of each dimension.
The data can also be clustered based on this representation into meaningful groups.

Our tri-geometry approach is applied to neuronal recordings and is used for exploratory analysis, interpretability and visualization of the data.
This organization is needed to identify latent variables that control the activity in the brain and to develop the automated infrastructure necessary to recover complex structures, with less external information and without expert guidance.
Our experimental results on neuronal recordings of head-fixed mice demonstrate the capability of isolating and filtering regional activities and relating them to specific stimuli and physical actions, and of automatically extracting pathological dysfunction.
Specifically neuronal groupings, temporal activity groupings and experimental condition scenarios are simultaneously extracted from the database, in a data-driven, model-free and network-oriented manner.

The remainder of the paper is organized as follows. 
Section~\ref{sec:related} briefly reviews related work regarding state-of-the-art methods in neuroscience data analysis.
In Section~\ref{sec:problem}, we formulate the problem. 
In Section~\ref{sec:quest3d}, the proposed methodology for tri-organization of trial-based data is presented, detailing the three components of our approach.
Section~\ref{sec:results} presents analysis of experimental neuronal data, in a motor foewpaw reach task in mice.

\section{Related Work}
\label{sec:related}
Current network analysis approaches in neuroscience can be divided into two main classes~\cite{Friston2013,Sporns2014}. 
The first class comprises methods, which aim to determine functional connectivity, defined in terms of statistical dependencies between measured elements (e.g., neurons or networks), by constructing direct statistical models from the data (e.g., Granger causality, transfer entropy, point process modeling and graph based methods~\cite{DinCheBre06,Schreiber00,TruEdeFel05,Sporns2014}. 
The second class of methods is often based on Latent Dynamical Systems (LDS), which accommodates effective connectivity characterizing the causal relations between elements through an underlying hidden dynamical system \cite{Friston2013,Shenoy2013,Archer2014}.
Non-linear and non-Gaussian extensions of the Kalman filter, contemporary sequential Monte Carlo methods and particle filters, have also been introduced in neuroscience~\cite{Wu04,Ahmadian2011}. 

\begin{figure*}[t!]
\centering{\includegraphics[width=0.88\linewidth]{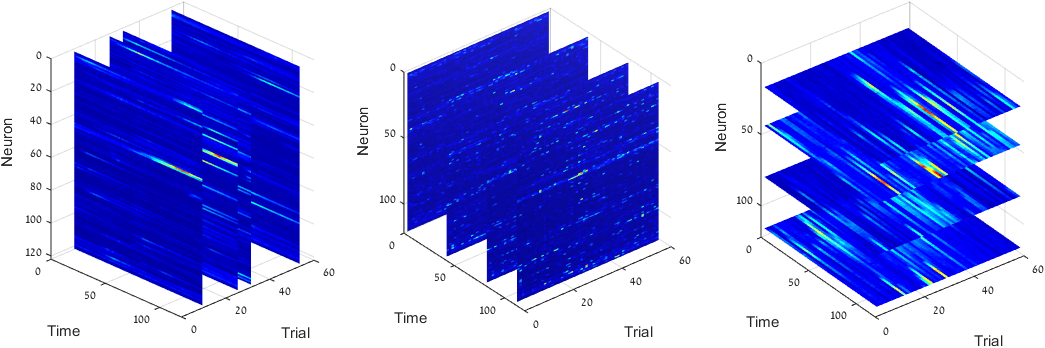}}
\caption{Visualization of 3D database.  The data is visualized here as 2D slices from multiple viewpoints: for several trials $\mathbf{X}_{\cdot \cdot T}$ (left), time frames $\mathbf{X}_{\cdot t \cdot}$ (center), and neurons $\mathbf{X}_{r \cdot \cdot}$(right).
The neuronal activity is represented by the intensity level of the image (blue – no activity, red – high activity). }
\label{fig:data3d_slices}
\end{figure*}
These methods share significant drawbacks, as they are mostly heuristic, providing only an approximation of a largely unknown system, and their quality is often hard to assess~\cite{Cunningham2014}. 
More importantly, they are all prone to the ``curse of dimensionality". 
On the one hand, designing a parametric generative model for truly complex high-dimensional data, such as neuronal/behavioral recordings, requires considerable flexibility, resulting in a model with a large number of tunable parameters. 
On the other hand, estimating a large number of parameters, and fitting a predefined class of dynamical models to high-dimensional data, is practically infeasible, thereby leading to poor data representations.
Our approach is better designed to deal with dynamical systems and aims to alleviate the shortcomings present in current analysis methods.
The proposed framework deviates from common recently used in neuroscience as it makes only very general smoothness assumptions, rather than postulating a-priori specific structural models.
In addition, we show that it takes into consideration the high dimensional spatio-temporal neuronal network structure. 

\section{Problem Formulation}
\label{sec:problem}
In the sequel we denote the three axes of the 3D data with a trial-based experiment in mind. 
However, our methodology can be applied to general 3D coordinates.
Consider data acquired in a set of fixed-length trials, composed of measurements from a large number of sensors (specifically neurons). 
Mathematically, we have a database $\mathbf{X}[r,t,T]$, depending on three variables. 
We collect at each neuron, or identified region of interest (ROI), denoted by $r$, a time series of the neuronal activity (e.g., fluorescence intensity levels in identified ROIs along time). 
In general trial-based data, this dimension corresponds to the sensors that acquire the data, such as in EEG~\cite{Talmon2015}.
On a short time scale, denoted by $t$, these time series can be viewed as a dynamic window profiling the neuron. 
These profiles vary on a long time scale as well, characterized by repetitive trials and denoted by $T$, and should be organized according to global trends and similarity between trials that are not necessarily consecutive. 
This database can be separately organized into a triple of geometries involving each variable, $r$, $t$, and $T$. 
However, the \emph{joint} organization of all three variables leads to an organization of dynamic neuronal activity regimes, using a global representation via the diffusion maps embedding~\cite{Coifman2006}. 

Let $\mathbf{X} \in \mathbb{R}^{n_r \times n_t \times n_T}$ be a rank-3 tensor, where $n_r$ is the number of neurons, $n_t$ is the number of time frames in an individual trial and $n_T$ is the number of trials.
A point $x[r,t,T]\in \mathbf{X}$ is indexed by the neuron $r$, the fast (short scale) time index $t$, and the trial (long scale time) index $T$. 
Note that although both $t$ and $T$ are coupled as indicating time, there is no assumption on a connection between the two.
Let $\mathbf{X}_{r\cdot\cdot} \in  \mathbb{R}^{n_t \times n_T}$ denote the two-dimensional matrix (slice) $\mathbf{x}[r,1\leq t \leq n_t, 1 \leq T \leq n_T]$ of all measurements for a given neuron $r$ throughout all trials.
In similar fashion, $\mathbf{X}_{\cdot t \cdot} \in  \mathbb{R}^{n_r \times n_T}$ is the 2D matrix of all measurements of all neurons, for a given time $t$ for all $n_T$ trials.
Finally,  $\mathbf{X}_{\cdot \cdot T} \in  \mathbb{R}^{n_r \times n_t}$ is the 2D matrix of all measurements of all neurons throughout a single trial $T$.
A visualization of a 3D dataset and examples of 2D slices in each dimension is presented in Fig.~\ref{fig:data3d_slices}.

Considering trial-based data, we assume the within-trial time index $t$ is smooth, and all trials are of the same length $n_t$.
It is easy to define neighbors in this dimension, as it is associated with a regular fixed-length grid.
We assume the trials follow a repetitive protocol, controlled by the experimenters, yet the trial indices $T$ are discrete, not necessarily contiguous and describe a longer span of time, i.e., trials occurring on different dates. 
Thus, the trial index $T$ while being associated with the notion of time which is supposedly smooth, does not imply that two consecutive indices are similar. 
In the experimental results in Sec.~\ref{sec:results} we show that trials are grouped logically based on behavioral similarity and not based on consecutive experiments.
A global trend in the organization of the trials is evident only when introducing a pathological inhibitor, which has a long term effect on the test subject. 
Finally, we assume that a neuron index $r$ may be assigned randomly to the neuron, therefore, it does not impose any smoothness or structure, and two consecutive indices do not imply any similarity between neurons.

Thus, although the trial-based measurements are organized as a 3D database so they are supposedly associated with a regular Euclidean grid, in practice the data suffers from non-uniform sampling, and consecutive indices do not indicate actual proximity as in time-series data (temporal smoothness) or a 2D image (spatial smoothness).
Thus, conventional analysis methods, such as multiscale representations via wavelets, are not straightforward in the given application.
In order to define a multiscale analysis of the data, it is necessary to be able to define neighborhoods and a sense of locality between points. 

The notations in this paper follow these conventions: matrices and tensors are denoted by bold uppercase, sets are denoted by uppercase calligraphic, and vectors are denoted by uppercase italic.

\section{Tri-geometry analysis}
\label{sec:quest3d}
Our analysis is based on the assumption that an underlying ``good'' organization of the data exists, such that under a permutation of the indices in each dimension of the data, the resulting tensor is smooth in the three dimensions.
Our aim is to recover this organization of the data, through a local to global processing of the data.
We begin with learning the hierarchal structure of the data in each dimension via partition trees, which convey local clustering of the data.
We then construct a new multiscale bi-tree metric for one dimension based on the coupled geometry between the other two dimensions.
Finally, the tree-based metric enables us to define an affinity between points from which we derive a global representation via manifold learning.
Thus, our approach does not treat each dimension separately, but introduces a strong coupling between the dimensions.
The three-phase organization of each dimension is carried out in an iterative procedure, where each dimension is organized in turn, based on the other two.

An advantage of our approach is that it is based on modular components.
We describe three methods fulfilling the motivation of each stage, but these methods can be replaced with others. For example, we propose flexible trees for the partition tree construction, but binary trees can be used instead.
We expand on the three components of our approach in detail in the following subsections.

\subsection{Partition trees and Flexible trees}
\label{sec:trees}
Following the assumption that under a proper organization the dataset is smooth, we aim to  build a fine-to-coarse set of neighborhoods for each point, by constructing partition trees in each dimension.
In the tri-geometric organization, the neighborhoods are 3D cubes.
Permuting the points in each dimension based on the constructed partition tree will recover the smooth structure respecting the coupling between the neurons and the time dimensions of the data.

Given a set of high-dimensional points, we construct a hierarchical partitioning of the points, defined by a tree.
In our setting, for each dimension, the set of points are the 2D slices of the data in that dimension.
Without loss of generality, we will define the partition trees in this section with respect to partitioning the neurons, but this process is performed in the remaining two dimensions as well. 

Let $\mathcal{X}_r = \{\mathbf{X}_{i \cdot\cdot}\}_{i=1}^{n_r}$  be the set of all 2D neuron slices.  
We define a finite partition tree $\mathcal{T}_r$ on $\mathcal{X}_r$ as follows.
The partition tree is composed of $L+1$ levels, where a partition of the points $\mathcal{P}_l$ is defined for each level $0 \leq l \leq L$.   
The partition $\mathcal{P}_l$ at level $l$ consists of $n(l)$ mutually disjoint non-empty subsets of indices in $\{1,...,n_r\}$, termed folders and denoted by $I_{l,i}$, $i\in\{1,...,n(l)\}$:
\begin{equation}
 \mathcal{P}_l = \{I_{l,1},I_{l,2},...,I_{l,n(l)}\}.
\end{equation}
Note that we define the folders on the indices of the data points and not on the points themselves.

The partition tree $\mathcal{T}_r$ has the following properties:
\begin{itemize}
\item The finest partition ($l = 0$) is composed of $n(0) = n_r$ singleton folders, termed the ``leaves'', where $I_{0,i} = \{i\}$.
\item The coarsest partition ($l= L$) is composed of a single folder, $I_{L,1} = \{1,...,n_r\}$, termed the ``root'' of the tree. 
\item The partitions are nested such that if $I \in \mathcal{P}_l$, then $I \subseteq J$ for some $J \in \mathcal{P}_{l+1}$, i.e., each folder at level $l-1$ is a subset of a folder from level $l$.
\end{itemize}
The partition tree is the set of all folders at all levels $\mathcal{T} = \{I_{l,i} \;\vert\; 0 \leq l \leq L,\; 1 \leq i \leq n(l)\}$, and the number of all folders in the tree is denoted by $\vert \mathcal{T} \vert$.

Given a dataset, there are many methods to construct a hierarchical tree, including deterministic, random, agglomerative and divisive~\cite{Gavish2010,Chi2014,Breiman2001}.
Partition trees can be constructed in a top-down or bottom-up approach.
In a top-down approach, the data are divided into few folders, then each of these folders is divided into sub-folders, and so on until all folders at the bottom level consist of only one point.
In a bottom-up approach, we begin with the lowest level of the tree, clustering the points into small folders. 
Then these folders are merged into larger folders at higher levels of the tree, until all folders are merged at the root of the tree.

A simple approach to bottom-up construction is a $k$-means based construction.
The leaves of the tree are clustered via $k$-means into $k$ folders.
Each folder is then assigned a centroid, and these centroids are then clustered again using $k$-means, with smaller $k$.
This process is repeated until all points are merged at the root with $k=1$.

More sophisticated approaches take into account the geometric structure and multiscale nature of the data by incorporating affinity matrices on the data, and manifold embeddings.
Gavish et al.~\cite{Gavish2010} propose constructing a partition tree via bottom-up hierarchical clustering, given a symmetric affinity matrix $\mathbf{W}$ describing a weighted graph on the dataset. 
Ankenman~\cite{Ankenman2014} proposed ``flexible trees'', whose construction requires an affinity matrix on the data, and is based on a low-dimensional diffusion embedding of the data, and not on the high-dimensional points. 
The advantage of this construction, which uses the embedding rather than the high-dimensional space is that distances between points in the diffusion space are meaningful and robust to noise, as opposed to high-dimensional Euclidean distances.
This tree construction enables us to integrate both the multiscale metric and the resulting global embedding.
Since our approach is based on an iterative procedure of all three components, the tree construction is refined from iteration to iteration.

Another important advantage of flexible trees is that there are relatively few levels and the level at which folders are joined is meaningful across the entire dataset.
Thus, the tree structure is logically multiscale and follows the structure of the data. 
This also reduces the computational complexity of the metric calculation.
The construction is controlled by a constant $\epsilon$ which affects the number of levels in the trees. 
Higher values of $\epsilon$ result in ``tall'' trees, while small values lead to flatter trees.

We briefly describe the flexible trees algorithm, given the set $\mathcal{X}_r$ and an affinity matrix on the neurons denoted $\mathbf{W}_r$.
For a detailed description see~\cite{Ankenman2014}. 
\begin{enumerate}
\item Input: The set of points $\mathcal{X}_r$, an affinity matrix $\mathbf{W}_r\in \mathbb{R}^{n_r \times n_r}$, and a constant $\epsilon$.
\item Init: Set partition $I_{0,i} = \{i\} \; \forall \; 1 \leq i \leq n_r$, set $l=1$.
\item Given an affinity on the data, we construct a low-dimensional embedding on the data~\cite{Coifman2006}.
\item \label{item:dist} Calculate the level-dependent pairwise distances $d^{(l)}(i,j) \; \forall \; 1 \leq i,j \leq n_r$ in the embedding space.
\item Set a threshold $\frac{p}{\epsilon}$, where $p=\textrm{median}\left(d^{(l)}(i,j)\right)$.
\item For each point $i$ which has not yet been added to a folder, find its minimal distance $d^{\min}(i)=\min_j\{d^{(l)}(i,j)\}$.
\begin{itemize}
\item If $d^{\min}(i)<\frac{p}{\epsilon}$, $i$ and $j$ form a new folder if $j$ also does not belong to a folder. 
If $j$ is already part of a folder $I$, then $i$ is added to that folder if $d^{\min}(i)<\frac{p}{\epsilon} 2^{-\vert I \vert + 1}$. 
Thus, the threshold on the distance for adding an element to an existing folder is divided by 2 for each added element.
\item If $d^{\min}(i) > \frac{p}{\epsilon}$, $i$ remains as a singleton folder.
\end{itemize}
\item \label{item:partition} The partition $P_l$ is set to be all the formed folders.
\item For $l>1$ and while not all points have been merged together in a single folder, steps \ref{item:dist})-\ref{item:partition}) are repeated. 
Instead of iterating over points, we iterate over all the folders $I_{l-1,i} \in P_{l-1}$. 
The distances between folders depend on the level $l$,  and on the points in the folder.
\end{enumerate}

The trees yield a hierarchical multiscale organization of the data, which then enables us to apply signal processing methods.
For example, we can apply non-local means to each point based on its neighborhood, to denoise the data, or multiscale analysis via tree based wavelets~\cite{Buades2005,Gavish2010,Ram}.
However, we aim at a global analysis of the data.
To this end, we define a bi-tree multiscale metric, which compares two points, based on their decomposition via the trees.

\subsection{Data-adaptive bi-tree multiscale metric} 
Applying manifold learning requires an appropriate metric between points.
As we cannot associate a sense of locality based on the indexing of the dimensions, we treat the data as vertices in a graph and develop a metric that is based on the multi-scale neighborhoods constructed in the partition tree. 
Given the partition trees in two of the dimensions, our aim is to define a distance $d$ between two 2D slices in the remaining dimension.
This distance should incorporate the multiscale nature of the data.

Following Leeb~\cite{Leeb2015}, we propose a tree-based metric in one dimension that incorporates the coupling of the other two dimensions.
For a two-dimensional matrix, Leeb~\cite{Leeb2015} defines a tree-based metric between two points in one dimension based on a partition tree in the other dimension.
We will present this metric in our context and then propose a new metric incorporating two partition trees in the case of a 3D dataset.

Consider a single 2D slice of the trial data $\mathbf{X}_{\cdot \cdot T}$, and the partition tree on the neurons $\mathcal{T}_r$.
A point $\mathbf{X}_{\cdot t_i T}$ in the time dimension  is a vector of length $n_r$, consisting of all the neuronal measurements at the time frame $t_i$ during a given trial $T$.
The tree metric $d_{\mathcal{T}_r}(\mathbf{X}_{\cdot t_i T},\mathbf{X}_{\cdot t_j T})$ between two times $t_i$ and $t_j$ within this trial, given the tree $\mathcal{T}_r$ is defined as
\begin{equation}
\label{eq:emd}
d_{\mathcal{T}_r}(\mathbf{X}_{\cdot t_i T},\mathbf{X}_{\cdot t_j T}) = \sum_{I \in \mathcal{T}_r} \vert m(\mathbf{X}_{\cdot t_i T} - \mathbf{X}_{\cdot t_j T}, I)\vert \omega(I),
\end{equation}
where $\omega(I)>0$ is a weight function, depending on the folder $I$.
The value $m(\mathbf{X}_{\cdot t_i T}, I)$ is the mean of vector $\mathbf{X}_{\cdot t_i T}$ in $I$:
\begin{equation}
m(\mathbf{X}_{\cdot t_i T}, I) = \frac{1}{\vert I \vert} \sum_{k\in I} \mathbf{X}[k,t_i,T], 
\end{equation}
where $\vert I \vert$ denotes the number of points in folder $I$, i.e., its cardinality.
The metric encompasses the ability to weight the data based on its multiscale decomposition since each folder is assigned a weight via $\omega$.
The weights can incorporate prior smoothness assumptions on the data, and also enable enhancing either coarse or fine structures in the similarity between points. 

We generalize this metric to a distance between 2D matrices, given two partitions trees.
We define this distance for the trial dimension, given trees on the time and neuron dimensions, but the same applies in the other dimensions as well.
Given a partition tree $\mathcal{T}_r$ on the neurons and a partition tree $\mathcal{T}_t$ on the time frames, the distance between two trials $T_i$ and $T_j$ is defined as
\begin{equation}
\label{eq:2demd}
d_{\mathcal{T}_r,\mathcal{T}_t}(\mathbf{X}_{\cdot \cdot i} , \mathbf{X}_{\cdot \cdot j}) = \sum_{\substack{I \in \mathcal{T}_r \\ J \in \mathcal{T}_t}} \vert m(\mathbf{X}_{\cdot \cdot i} - \mathbf{X}_{\cdot \cdot j}, I \times J)\vert \omega(I,J),
\end{equation}
where $\omega(I,J)>0$ is a weight function depending on folders $I \in \mathcal{T}_r $ and $J \in \mathcal{T}_t$.
We term this distance a bi-tree metric.
The value $m(\mathbf{X}_{\cdot \cdot i}, I \times J)$ is the mean value of a matrix $\mathbf{X}_{\cdot \cdot i}$ on the bi-folder $I \times J = \{ (k,n) \vert k\in I, n\in J \}$:
\begin{equation}
m(\mathbf{X}_{\cdot \cdot i}, I \times J) = \frac{1}{\vert I \vert \vert J \vert } \sum_{k \in I, n\in J} \mathbf{X}[k,n,i],
\end{equation}
i.e., for a given trial $T$, we are averaging the sub-matrix of the 2D slice  $\mathbf{X}_{\cdot \cdot i}$ defined by the subset of neurons in $I$ and the subset of time frames in $J$.

We present a new interpretation of the tree-based metrics~(\ref{eq:emd}) and~(\ref{eq:2demd}).
These metrics are equivalent to the $l_1$ distance between points, after applying a multiscale transform to the data, where the tree metric~(\ref{eq:emd}) corresponds to a 1D transform and the bi-tree metric~(\ref{eq:2demd}) corresponds to a 2D transform.
For the sake of simplicity we begin with describing the 1D transform in the case of a single 2D slice of the trial data $\mathbf{X}_{\cdot \cdot T}$, and then generalize to the 2D transform. 

The partition tree $\mathcal{T}_r$ can be seen as inducing a multiscale decomposition on the data, via the construction of a data-adaptive filter bank.
Define the filter $f_I \in \mathbb{R}^{n_r}$ as
\begin{equation}
\label{eq:filter}
 f_I=\frac{\omega(I)}{\vert I \vert} \mathds{1}_I,
\end{equation}
such that $\mathds{1}_I$ is the indicator function on the neurons $i\in\{1,...,n_r\}$ belonging to folder $I \in \mathcal{T}_r$.
For each filter we calculate the inner product between the filter $f_I$ induced by folder $I$ and the measurement vector $\mathbf{X}_{\cdot t T} \in \mathbb{R}^{n_r}$, yielding a scalar coefficient $g_I$: 
\begin{equation}
\label{eq:coef}
\begin{split}
 g_I(\mathbf{X}_{\cdot t T}) & = \langle f_I,\mathbf{X}_{\cdot t T} \rangle \\ 
 & = \frac{\omega(I)}{ \vert I\vert} \sum_{k\in I} \mathbf{X}[k,t,T]=m(\mathbf{X}_{\cdot t T},I)\omega(I).
 \end{split}
\end{equation}

The tree $\mathcal{T}_r$ defines a multiscale transform by applying filter bank $f_{\mathcal{T}_r} = \{f_I\}_{I\in\mathcal{T}_r}$ to the measurements vector $\mathbf{X}_{\cdot t T}$, resulting in the set of coefficients $g_{\mathcal{T}_r} = \{g_I\}_{I\in\mathcal{T}_r}$. 
The filters of each level $l$ of the tree output $n(l)$ coefficients, such that $g_{\mathcal{T}_r} : x \mapsto \mathbb{R}^{\vert \mathcal{T}_r \vert}$.
This is demonstrated in Fig.~\ref{fig:transform}. 
In the middle, a 2D slice $\mathbf{X}_{\cdot t T}$ is viewed as a 2D matrix and on the left is a partition tree $\mathcal{T}$ defined on the rows of the matrix.
We assume that the rows of the matrix have been permuted so they correspond with leaves of the tree (level 0).
In applying the transform $g_\mathcal{T}$, each folder $I$ defines an element in the new vector $g_\mathcal{T}(X_{\cdot i})$ (right), proportional to the average of the entries in the original vector $(X_{\cdot i})$ on the support defined by the folder $I$.
The new entries in the vector are colored according to the corresponding folders in the tree.

\begin{theorem}
 Given a partition tree on the neurons $\mathcal{T}_r$, the tree metric~(\ref{eq:emd}) between two times $t_j$ and $t_j$ for a given trial $T$ is equivalent to the $l_1$ distance between the multiscale transform defined by the tree and applied to the two vectors:
 \begin{equation}
  d_{\mathcal{T}_r}(\mathbf{X}_{\cdot t_i T},\mathbf{X}_{\cdot t_j T}) = \Vert g_{\mathcal{T}_r}(\mathbf{X}_{\cdot t_i T}) - g_\mathcal{T}(\mathbf{X}_{\cdot t_j T}) \Vert_1.
\end{equation}
\end{theorem}

\begin{proof}
\begin{align}
  d_{\mathcal{T}_r}(\mathbf{X}_{\cdot t_i T} &, \mathbf{X}_{\cdot t_j T}) 
   =  \sum_{I \in \mathcal{T}_r} \vert m(\mathbf{X}_{\cdot t_i T} - \mathbf{X}_{\cdot t_j T}, I)\vert \omega(I) \nonumber \\
  & =  \sum_{I \in \mathcal{T}_r} \vert m(\mathbf{X}_{\cdot t_i T},I)\omega(I) - m(\mathbf{X}_{\cdot t_j T}, I) \omega(I)\vert \nonumber \\ 
  & = \sum_{I \in \mathcal{T}_r} \vert g_I(\mathbf{X}_{\cdot t_i T}) - g_I(\mathbf{X}_{\cdot t_j T}) \vert \nonumber \\
  & = \sum_{n=1}^{\vert \mathcal{T}_r \vert} \vert g_{\mathcal{T}_r}(\mathbf{X}_{\cdot t_i T})[n] - g_{\mathcal{T}_r}(\mathbf{X}_{\cdot t_j T})[n] \vert \nonumber \\
  & = \Vert g_{\mathcal{T}_r}(\mathbf{X}_{\cdot t_i T}) - g_{\mathcal{T}_r}(\mathbf{X}_{\cdot t_j T}) \Vert_1. 
\end{align}
\end{proof}
This result can be generalized to a multiscale 2D transform applied to 2D matrices as in our setting.
Define the 2D filter $f_{I\times J}$ by:
\begin{equation}
\label{eq:filter2d}
 f_{I\times J}=\frac{\omega(I,J)}{\vert I \vert\vert J \vert} \mathds{1}_I \otimes  \mathds{1}_J^T,
\end{equation} 
where $\otimes$ denotes the Kronecker product between the two indicator vectors.
Then the elements of the 2D matrix $g_{\mathcal{T}_r,\mathcal{T}_t} \in \mathbb{R}^{\vert \mathcal{T}_r \vert \times \vert \mathcal{T}_t \vert }$ are the coefficients obtained from applying the 2D filter bank $f_{\mathcal{T}_r,\mathcal{T}_t} = \{f_{I\times J}\}_{I \in \mathcal{T}_r,J \in \mathcal{T}_t}$ defined by the bi-tree $\mathcal{T}_r \times \mathcal{T}_t$.
\begin{corollary}
 The bi-tree metric~(\ref{eq:2demd}) between two matrices given a partition tree $\mathcal{T}_r$ on the neurons and a partition tree $\mathcal{T}_t$ on the time frames is equivalent to the $l_1$ distance between the 2D multiscale transform of the two matrices:
\begin{equation}
d_{\mathcal{T}_r,\mathcal{T}_t}(\mathbf{X}_{\cdot \cdot i} , \mathbf{X}_{\cdot \cdot j}) = \Vert g_{\mathcal{T}_r,\mathcal{T}_t} (\mathbf{X}_{\cdot \cdot i}) - g_{\mathcal{T}_r,\mathcal{T}_t}(\mathbf{X}_{\cdot \cdot j}) \Vert_1.
\end{equation}
\end{corollary}

\begin{figure}[t]
\centering{\includegraphics[width=0.99\linewidth]{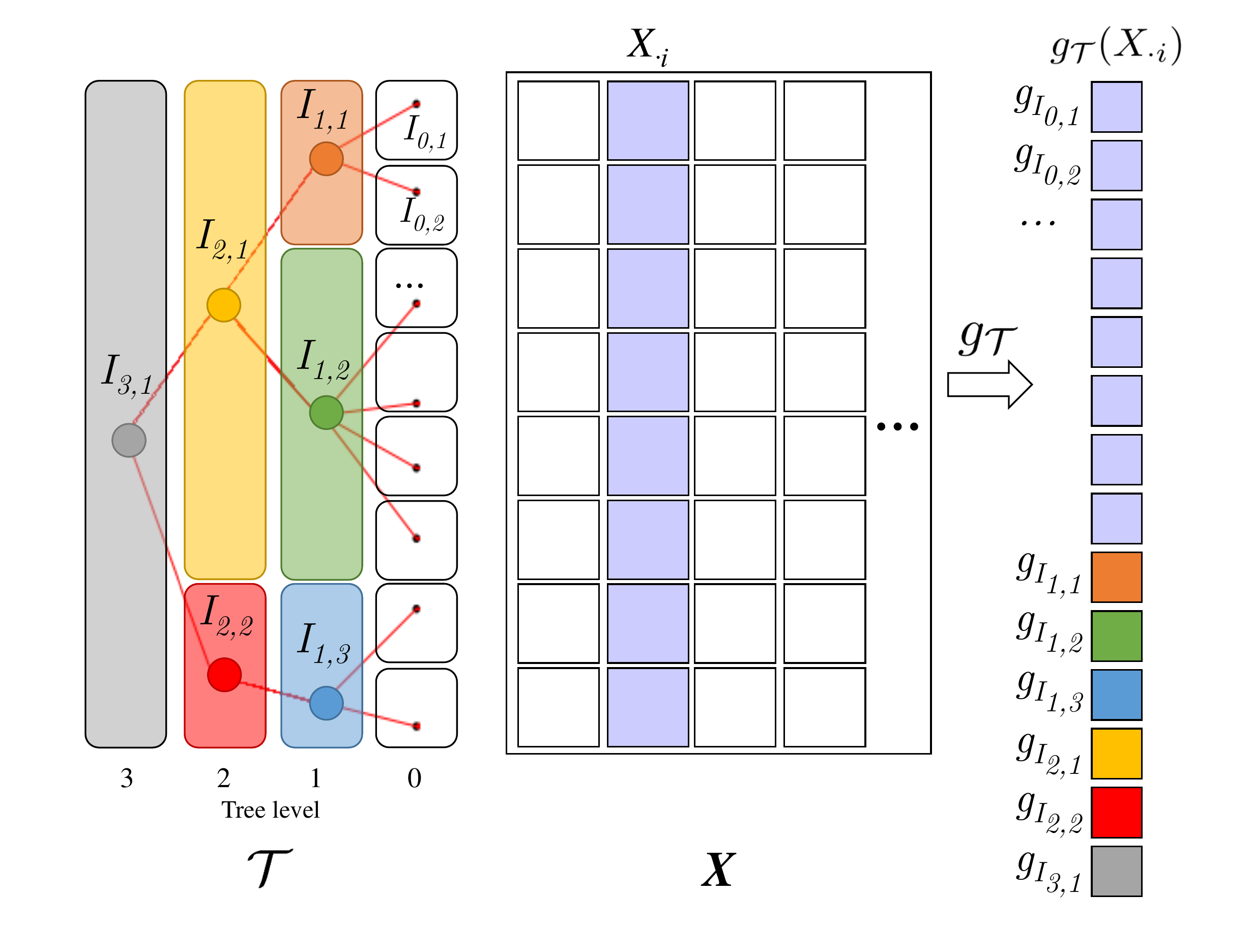}}
\caption{Multiscale 1D tree-transform applied to a 2D slice from Fig.~\ref{fig:data3d_slices}, viewed here as a 2D matrix (middle).
On the left is a given partition tree $\mathcal{T}$ on the rows of the 2D matrix, and we assume the rows have been permuted so the leaves of the tree correspond to the rows.
The partition tree $\mathcal{T}$ defines a multiscale transform on the columns of the matrix $X_{\cdot i}$, resulting in new vectors $g_\mathcal{T}(X_{\cdot i})$. 
In applying the transform $g_\mathcal{T}$, the entries in $X_i$ corresponding to each folder in the tree, are averaged and weighted according to~(\ref{eq:coef}).
This yields new scalar coefficients which form the output vector $g_\mathcal{T}(X_{\cdot i})$ (right).
For visualization, each new entry $g_I$ is colored by the corresponding folder $I$ in the tree.
}
\label{fig:transform}
\end{figure}
This interpretation of the metric as the $l_1$ distance between multiscale transforms has two computational advantages.
First, given large datatsets, it is inefficient to calculate full affinity matrices on the points, and instead sparse matrices are used by finding $k$-nearest neighbors of each point.
Thus, we can apply the multiscale transform to our data, yielding a new feature vector for each point, and then apply approximate nearest-neighbor search for the $l_1$ distance to the new vectors~\cite{Arya:1998,Yi2000}.
Second, we can relax the $l_1$ norm to other norms such as $l_2$ or $l_\infty$.
In future work, we intend to establish the properties of this transform and its application to other tasks.

Note that we claimed that regular metrics are inappropriate in processing the data due to its high-dimensionality in each dimension of the 3D dataset, i.e., each 2D slice of the data contain a large number of elements.
This interpretation of the metric via the transform yields that the proposed metric is equivalent to the $l_1$ distance between vectors/matrices of even higher-dimensionality, supposedly contradicting our aim for a good metric.
However, due to encompassing weights on the folders, the effective size of the new vectors is smaller than the original dimensionality, as the weights are chosen such that they rapidly decrease to zero based on the folder size.

We note that by using full binary trees in each of the two dimensions, the output of applying the multiscale transform is similar to that of applying the Gaussian pyramid representation, popular in image processing~\cite{Burt1983}, to each 2D matrix $\mathbf{X}_{\cdot \cdot i}$, $1 \leq i \leq n_T$.
Instead of applying the $5 \times 5$ Gaussian filter proposed by Burt and Adelson, our transform applies a $2 \times 2$ averaging filter, weighted by $\omega(I,J)$, and the resolution at each level will be reduced by 2 as in the Gaussian pyramid.

\paragraph*{Relationship to EMD}
The Earth mover's distance (EMD) is a metric used to compare probability distributions or discrete histograms, and is popular in computer vision\cite{Rubner1998}. 
It is fairly insensitive to perturbations since it does not suffer from the fixed binning problems of most distances between histograms. 
EMD quantifies the difference between the two histograms as the amount of mass one needs to move (flow) between histograms, with respect to a ground distance, so they coincide. 
In its discrete form, the EMD between two normalized histograms $h_1$ and $h_2$ is defined as the minimal total ground distance ``traveled'' weighted by the flow:
\begin{align*}
\textrm{EMD}(h_1,h_2)  = \min \sum_{i,j}g_{ij}d_{ij} \\
 \textrm{s.t.}  \sum_i g_{ik} - \sum_j g_{kj} = h_1(k)-h_2(k),
\end{align*}
where $d_{ij}\geq 0$ is the ground distance, and $g_{ij}$ is the flow from bin $i$ to bin $j$.

It was shown~\cite{Leeb2013} that a proper choice of the weight $\omega(I)$ makes the tree metric~(\ref{eq:emd}) equivalent to EMD, i.e., the ratio of EMD to the tree-based metric is always between two constants.
The proof follows the Kantorovich-Rubinstein theorem regarding the dual representation of the EMD problem.
The weight $\omega(I)$ in~(\ref{eq:emd}) is chosen to depend on the tree structure:
\begin{equation}
\label{eq:weight}
\omega(I) = \left(\frac{\vert I \vert}{M}\right)^{\beta+1},
\end{equation}
where $\beta$ weights the folder by its relative size.
Positive values of $\beta$ correspond to higher weights on coarser scales of the data, whereas negative values emphasize differences in fine structures in the data.
For trees with varying-sized folders, unlike a balanced binary tree, $\beta$ helps to normalize the weights on folders.
For $\beta=0$, the filter $f_I$ defined in~(\ref{eq:filter}) is a uniform averaging filter whose support is determined by $I$.
In EMD the histograms are associated with a fixed grid and bins quantizing this grid.
In our setting, where the data does not follow a fixed grid, the folders take the place of the bins, and by incorporating their multiscale structure via the weights, they can be seen as bins of varying sizes on the data.

Shirdhonkar and Jacobs~\cite{Shirdhonkar2008} proposed a wavelet-based metric (wavelet EMD) shown to be equivalent to EMD.
The wavelet EMD is calculated as the weighted $l_1$ distance between the wavelet coefficients of the difference between the two histogram.
Following~\cite{Shirdhonkar2008}, Leeb~\cite{Leeb2013} proposed a second metric based on the $l_1$ distance between the coefficients of the difference of distributions expanded in the tree-based Haar-like basis~\cite{Gavish2010}, which was also  shown to be equivalent to EMD.
Our interpretation of the metric~(\ref{eq:emd}) as an $l_1$ distance between a multi-scale filter bank applied to the data, simplifies the calculation even more as it does not require calculating the Haar-like basis defined by the tree, and instead requires only low-pass averaging filters on the support of each folder.
This generalizes the wavelet EMD~\cite{Shirdhonkar2008}, to high-dimensional data that is not restricted to a Euclidean grid.

For the bi-tree metric~(\ref{eq:2demd}), the weight on a bi-folder $I \times J$ can be chosen in an equivalent manner to~(\ref{eq:weight}) as
\begin{equation}
\label{eq:omega2d}
\omega(I,J) = \left(\frac{\vert I \vert}{n_r}\right)^{\beta_r+1} \left(\frac{\vert J \vert}{n_t}\right)^{\beta_t+1},
\end{equation}
where $\beta_r$ weights the bi-folder $I \times J$ based on the relative size of folder $I\in\mathcal{T}_r$ and $\beta_t$ weights the bi-folder based on the relative size of $J \in \mathcal{T}_t$.
The values should be set according to the smoothness of the dimension and whether we intend to enhance coarse or fine structures in the data.

\subsection{Global Embedding}
The intrinsic global representation of the data is attained by an integration process of local affinities, often termed ``diffusion geometry". 
Specifically, the encoding of local variability and structure from different locations (e.g., cortical regions, or trials) is aggregated into a single comprehensive representation through the eigendecomposition of an affinity kernel~\cite{Coifman2006}. 
This global embedding  preserves local structures in the data, thus enabling us to exploit the fine spatio-temporal variations and inter-trial variability typical of biological data, in contrast to other methods based on averaging and smoothing the data~\cite{Pfau2013}. 

Given the bi-tree multiscale distance between two points (\ref{eq:2demd}), we can construct an affinity on the data along each dimension. 
We choose an exponential function, but other kernels can be considered, dependent on the application.
Without loss of generality, we describe the embedding calculation with respect to the dimension of the neurons, but this procedure is applied to the time and trials as well, within our iterative framework. 
Given the multiscale distance $d_{\mathcal{T}_t,\mathcal{T}_T}(\mathbf{X}_{i \cdot \cdot}, \mathbf{X}_{j \cdot \cdot})$ between two neurons $r_i$ and $r_j$, the affinity is defined as:
\begin{equation}
 a(r_i, r_j) = \exp\{ - d_{\mathcal{T}_t,\mathcal{T}_T}(\mathbf{X}_{i \cdot \cdot}, \mathbf{X}_{j \cdot \cdot}) / \sigma_r \},
\end{equation}
where $\sigma_r$ is a scale parameter, and depends on the current considered dimension of the 3D data, i.e., each dimension uses a different scale in its affinity.
Typically, $\sigma_r$ is chosen to be the mean of distances within the data.
The exponential function enhances locality, as points with distance larger than $\sigma_r$ have negligible affinity. 

The affinity is used to calculate a low-dimensional embedding of the data, using manifold learning techniques, specifically diffusion maps~\cite{Coifman2006}. 
Defining an affinity matrix $\mathbf{A}[i,j] = a(r_i, r_j),\; \mathbf{A} \in \mathbb{R}^{n_r \times n_r}$, we derive a corresponding row-stochastic matrix by normalizing its rows:
\begin{equation}
\mathbf{P} = \mathbf{D}^{-1}\mathbf{A}, 
\end{equation}
where $\mathbf{D}$ is a diagonal matrix whose elements are given by $\mathbf{D}[i,i] = \sum_j \mathbf{A}[i,j]$.
The eigendecomposition of $\mathbf{P}$ yields a a sequence of positive decreasing eigenvalues: $1 = \lambda_0\geq\lambda_1\geq ...$, and right eigenvectors $\{\psi_\ell\}_\ell$.
Retaining only the first $d$ eigenvalues and eigenvectors, the mapping $\Psi_r$ embeds the data set $X$ into the Euclidean space $\mathbb{R}^{d}$:
\begin{equation}
\label{eq:diffusion_map}
\Psi_r:\mathbf{X}_{r_i \cdot \cdot}\rightarrow \big( \lambda_1\psi_1(i), \lambda_2\psi_2(i),..., \lambda_{d}\psi_{d}(i)\big)^T.
\end{equation}
Note that for simplicity of notation we omit denoting the eigenvalues and eigenvectors by the relevant dimension $r$, $t$ or $T$.
The embedding provides a global low-dimensional representation of the data, which preserves local structures.
Euclidean distances in this space are more meaningful than in the original high-dimensional space, as they have been shown to be robust to noise.
The flexible tree construction is based on the embedding for these reasons.
Finally, the embedding integrates the local connections found in the data into a global representation, which enables visualization of the data, reveals overlying temporal trends, organizes the data into meaningful clusters, and identifies outliers and singular points.
For more details on diffusion maps, see~\cite{Coifman2006}.

\subsection{Algorithm}
\label{sec:implement}
\begin{algorithm}[t]
\caption{Hierarchical tri-geometric analysis}
\label{alg:3d}
\algsetup{indent=2em}
\begin{algorithmic}[1]
\INIT
\INPUT 3D data matrix $\mathbf{X}$
\STATE Starting with the neuron dimension $r$
\STATE \label{step:aff} \hspace{0.5cm} Calculate initial affinity matrix $\mathbf{A}_r^{(0)}(i,j)$ 
\STATE  \label{step:embed} \hspace{0.5cm} Calculate initial neuron embedding $\Psi_r^{(0)}$.
\STATE \label{step:tree} \hspace{0.5cm} Calculate initial flexible tree  $\mathcal{T}_r^{(0)}$. 
\STATE \label{step:time} For time dimension $t$ repeat steps \ref{step:aff}-\ref{step:tree} and obtain $\mathcal{T}_t^{(0)}$.
\ANALYSIS
\INPUT Flexible trees $\mathcal{T}_r^{(0)}$ and $\mathcal{T}_t^{(0)}$
\FOR{$n \geq 1$}
\STATE \label{step:dist} Calculate multiscale bi-tree distance between two trials $d(T_i,T_j)=d_{\mathcal{T}_r^{(n-1)},\mathcal{T}_t^{(n-1)}}(\mathbf{X}_{\cdot \cdot i}, \mathbf{X}_{\cdot \cdot j})$
\STATE  Calculate trial affinity matrix \\ $\mathbf{A}_T^{(n)}(i,j) = \exp\left\{-d(T_i,T_j) / \sigma_T \right\}$
\STATE \label{step:embed2} Calculate trial embedding $\Psi_T^{(n)}$
\STATE \label{step:tree2} Calculate flexible tree on the trials $\mathcal{T}_T^{(n)}$. 
\STATE For the neuron dimension $r$, repeat steps \ref{step:dist}-\ref{step:tree2}, given the trees on the time and trials, $\mathcal{T}_t^{(n-1)}$ and $\mathcal{T}_T^{(n)}$ respectively, and obtain $\mathcal{T}_r^{(n)}$.
\STATE For the time dimension $t$, repeat steps \ref{step:dist}-\ref{step:tree2}, given the trees on the trials and neurons, $\mathcal{T}_T^{(n)}$ and $\mathcal{T}_r^{(n)}$ respectively, and obtain $\mathcal{T}_t^{(n)}$.
\ENDFOR
\end{algorithmic}
\end{algorithm}
Our iterative analysis algorithm composing all three components (tree construction, metric building, embedding) is summarized in Algorithm~\ref{alg:3d}.
Note that the order in which the dimensions are processed is arbitrary, and it affect the final results. 
In addition, since the algorithm is iterative and each component relies on the previous components, an initialization is required.
Specifically, calculation of the bi-tree metric for one dimension requires that partition trees be calculated on the other two dimensions.

One option is to calculate an initial affinity matrix based on a general distance such as the Euclidean distance or cosine similarity.
Here, we use the cosine similarity:
\begin{equation}
\label{eq:cos}
 a^{\cos}(r_i,r_j) = \frac{\sum_{t,T} x[i,t,T]x[j,t,T]}{\sqrt{\sum_{t,T} \left(x[i,t,T] \right)^2} \sqrt{\sum_{t,T} \left( x[j,t,T] \right)^2}}.
\end{equation}
Note that although the affinity is supposedly between two matrices, effectively it is equivalent to reshaping the matrices as 1D vectors and calculating the affinity using 1D distances.
In other words, a general affinity does not take into account the 2D structure of the slices of the 3D data, in contrast to our bi-tree metric.
In addition, these distances are uninformative, as the data are extremely high-dimensional. 
For example, in each dimension of the dataset in the experimental results in Sec.~\ref{sec:results}, the dimension of the measurements is of order $10^4$.

Given the initial affinity, an embedding and flexible tree are calculated for the neuron dimension $r$ (steps \ref{step:embed}-\ref{step:tree}). 
This is then repeated for the time dimension (step~\ref{step:time}).
One second option is to initialize the partition tree for the time dimension to be a binary tree, since the intra-trial time $t$ is a smooth variable. 

Given the trees in two of the dimensions, we can calculate the multiscale metric~(\ref{eq:2demd}) in the trial dimension $T$ (step~\ref{step:dist}).
A corresponding embedding and flexible tree are then calculated (steps \ref{step:embed2}-\ref{step:tree2}).
We now have a partition tree in each dimension, so we continue in an iterative fashion, going over each dimension and calculating the multiscale metric, diffusion embedding and flexible tree in each iteration, based on the other two dimensions.
The resulting output of the algorithm can be used to analyze the data both in terms of its hierarchical structure and through visualization of the embedding.
Furthermore, each dimension can be organized by calculating a smooth trajectory in its embedding space.
This yields a permutation on the indices of the given dimension. 
Permuting all three dimensions recovers the smooth structure of the data, respecting the coupling between the neurons and the time dimensions of the data.
Python code implementing Algorithm~\ref{alg:3d} will be released open-source on publication.

\section{Results}
\label{sec:results}

\subsection{Experimental Setup}
\begin{figure}[t!]
\centering
\includegraphics[width=0.75\linewidth]{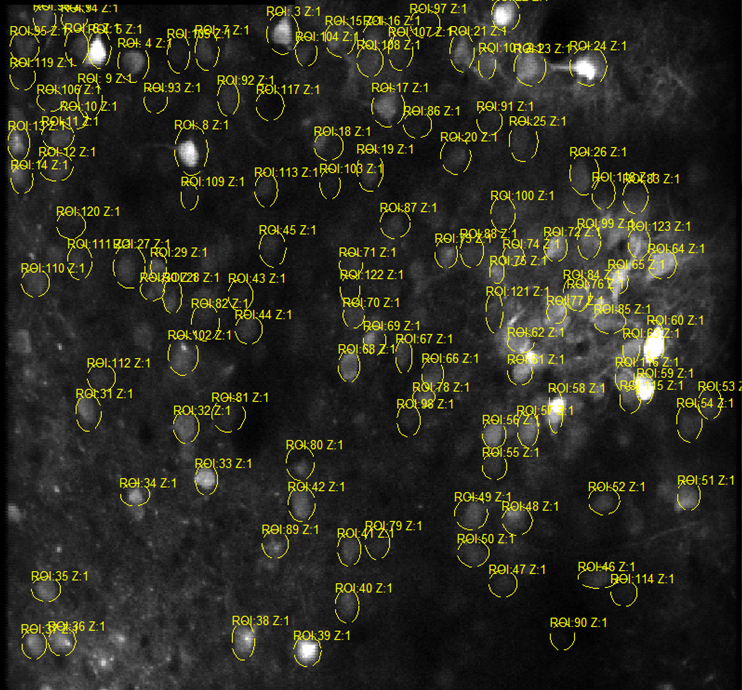}
\caption{Two-photon imaging in the primary motor cortex (M1). The neuronal measurements are gathered into regions of interest (ROIs), consisting of ellipses, and preprocessed as in (\ref{eq:F})-(\ref{eq:dFF}).}
\label{fig:exp}
\end{figure}
Our experimental data consists of repeated trials of a complex motor forepaw reach task in awake mice.
The animals were trained to reach for a food pellet upon hearing an auditory cue~\cite{Whishaw1992}. 
This complex and versatile task exploits the capability of rodents to use their forepaw very similarly to distal hand movements in primates~\cite{Whishaw1992}.
The hand reach task is typically learnt by mice over a period of few weeks to become ``experts" (success rate of $\sim 70\%-80\%$ after training over 2-3 weeks). 

Neuronal activity in the motor cortex during task performance was measured using two photon in-vivo calcium imaging with the recently developed genetically encoded indicators (GECIs)~\cite{Chen2013}. 
In addition the network was silenced using DREADDS~\cite{Rogan2011}, which was activated using intraperitoneal (IP) injection of the inert agonist (clozapine-N-oxide (CNO).
The analyzed neuronal measurements are of optical calcium fluorescent activity collected from a large population of identified neurons from cortical regions of interest, acquired using two photon microscopy imaging (Fig.~\ref{fig:exp}).
In conjunction, high-resolution behavioral recordings of the subject are acquired using a camera (400 Hz).
This serves to label the time frames and to determine whether the subject performed the task successfully during the trial.

The fluorescent measurements are manually grouped into elliptical regions of interest (ROIs) (Fig.~\ref{fig:exp}), and preprocessing is applied as follows. 
The spatial average fluorescence of each ROI $k$ per time frame $t$ in a single trial is
\begin{equation}
\label{eq:F}
F_k(t) = \frac{1}{\vert \textrm{ROI}_k \vert}\sum_{i,j\in \textrm{ROI}} I[i,j,t], 
\end{equation}
where $I$ is the florescence image, $i$ and $j$ are the pixel row and column indices in the image, respectively, and $\vert \textrm{ROI}_k \vert$ is the area of the $k$-th ROI.
The baseline florescence for ROI $k$ in a single trial $T$ is calculated using a subset of time frames $S_k$ corresponding to the florescent averages $F_k(t)$ with the $10\%$ lowest values $\bar{F}_k = \sum_{t\in S_k} F_k(t)$. 
Finally, the neuron measurement at each time frame $\mathbf{X}[k,t,T]$ is set using $\frac{\Delta F}{F}$:
\begin{equation}
\label{eq:dFF}
\mathbf{X}[k,t,T] = \frac{F_k(t) - \bar{F}_k}{ \bar{F}_k}.
\end{equation}
For simplicity, we refer to the ROIs as neurons in our analysis.

\subsection{Data}
We focus on neuronal measurements from the primary motor cortex region (M1), taken from a specific mouse in a single day of experimental training sessions. 
The data is composed of 59 consecutive trials, where the first 19 trials are considered ``control" followed by 40 trials in which the activity of the somatosensory region was silenced by injection of CNO, thus activating DREADDS. 
Each trial lasts 12 seconds, during which activity in 121 neurons is measured for 119 time frames. 
Thus, the data can be seen as 3-dimensional, measuring a vector of neurons at each time frame within each trial.
The data is visualized as 2D slices for several neurons, time frames and trials in Fig.~\ref{fig:data3d_slices}. 
The time frame (1-119) within the trial is a local time scale, and the trial index represents a global time scale (1-59). 

Along with neuron measurements, we also have binary data labeling an event for each time and trial (Fig.~\ref{fig:label}). 
The labeling is performed using a modified version of the machine learning based JAABA software, annotating discrete behavioral events~\cite{Kabra2013}.
There are 11 labeled events that provide additional prior information helpful in verifying our analysis. 
An auditory cue (``tone'' event) is activated after 4 seconds (frames 40-42) and the food pellet comes to position at 4.4 seconds (frames 44-46).
The ``tone" event is typically followed by either a successful ``grab'' event and ``at mouth'' event, which lasts until the end of the trial, or by a several failed ``grab'' events and then labeled as a ``miss" event, i.e., the subject failed to grab the food pellet and bring it to its mouth.

\begin{figure}[t!]
\centering{\includegraphics[width=0.95\linewidth]{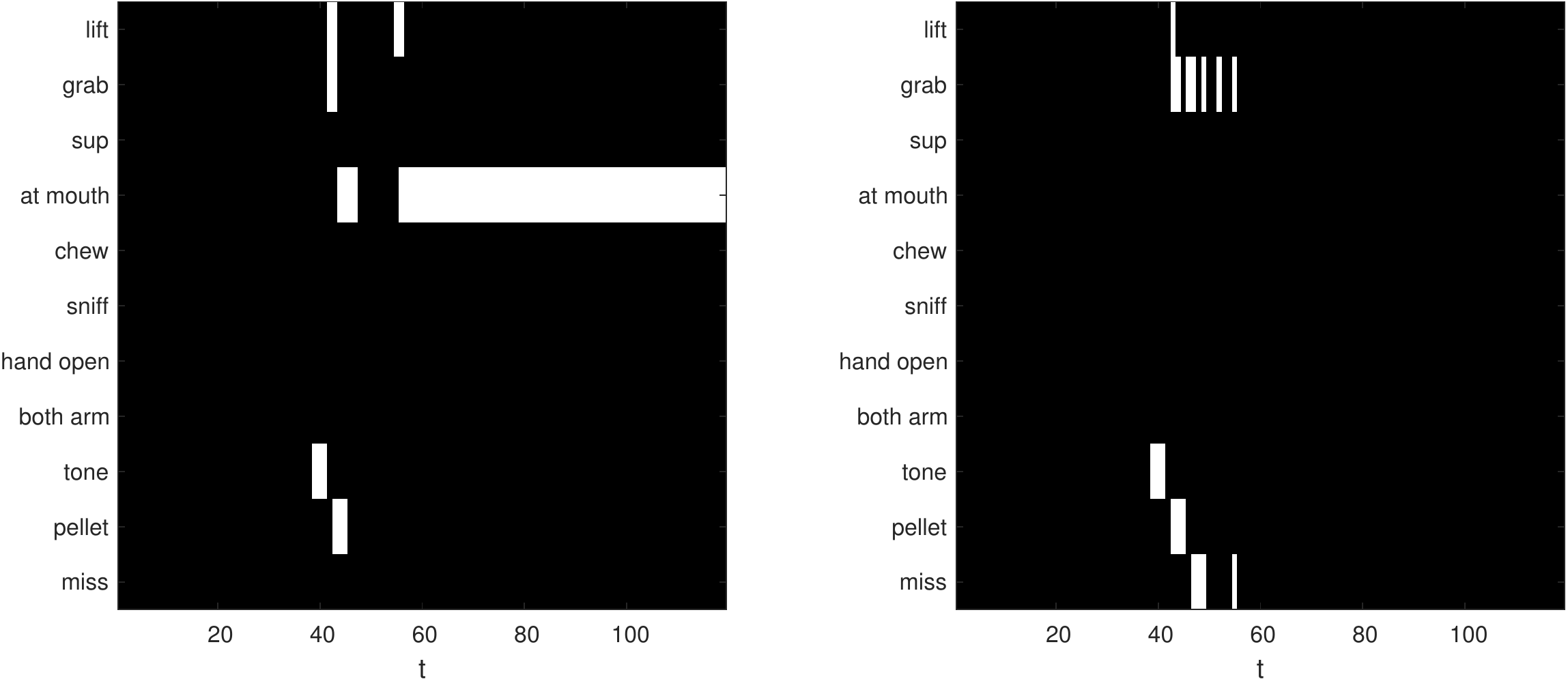}}
\caption{Binary event labels for two trials. (left) Successful trial in which the subject grabs and eats the food pellet. (right) Failure in which the subject makes several failed attempts to grab the food.} 
\label{fig:label}
\end{figure}

The control data consists of 19 trials, 11 of which were successful, i.e., the mouse managed to grab and eat the food pellet.
After 19 trials, CNO was injected IP to silence the sensory cortex (S1), which sends feedback information to M1. 
The next 40 trials, referred to as ``silencing trials'' included 10 successful trials. 
During these trials, the behavior of the mouse changes, demonstrated by a decrease in ``at mouth" (chewing) events and an increase in ``miss" events (in which the mouse does not manage to grab the food). 
Note that not all silencing trials are ``miss" and not all control trials are successful.

\subsection{Tri-geometric analysis}
The activity of the neurons is such that they are correlated at certain times, but completely unrelated at others, and certain neurons are sensitive to the auditory trigger, whereas others completely disregard it.
The goal is to automatically extract co-active communities of neurons, as they relate to the activity of the mouse.
We first analyze all 59 trials together, using Algorithm~\ref{alg:3d}. 
For the weights~(\ref{eq:omega2d}) used in the mutilscale metric~(\ref{eq:2demd}), we choose $\beta_r=1,\beta_t=1,\beta_T=0$. 
We describe in the following how the analysis is used to derive meaningful results for each dimension.

\begin{figure*}[th]
\centering{\includegraphics[width=0.9\linewidth]{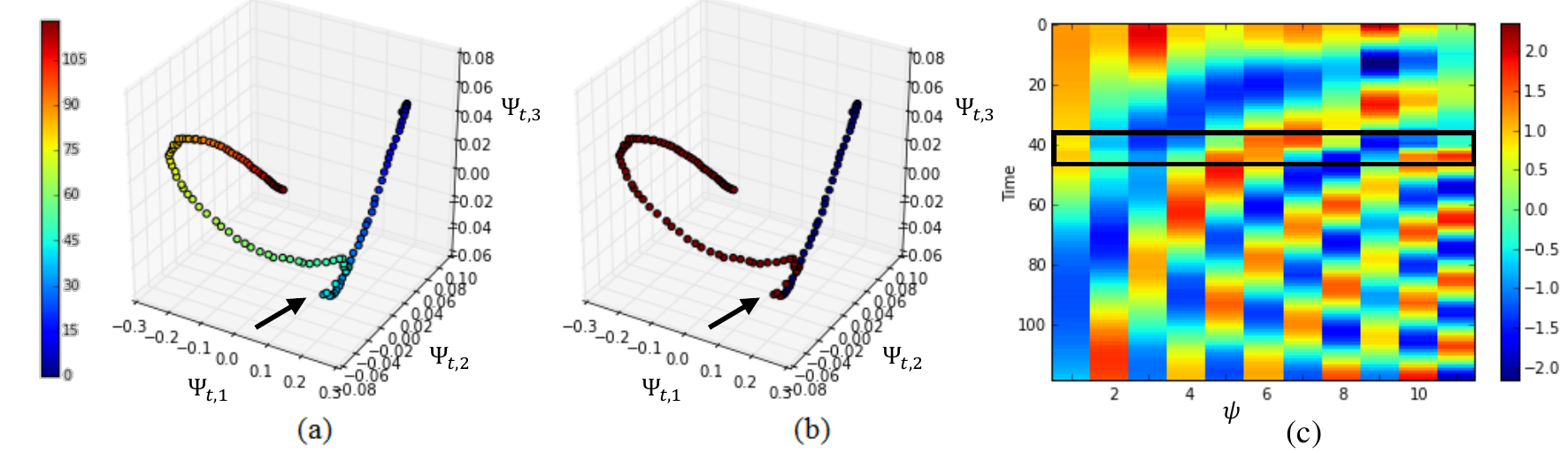}}
\caption{Embedding of time frames. (a-b) 3-dimensional embedding of all the 2D time frame slices (as in Fig.~\ref{fig:data3d_slices}(center), constructed by our tri-geometry analysis, where each time sample ($t\in\{1,...,119\}$) is a 3D point. 
In (a) the points are colored by the time frame index, and in (b) they are colored according to pre-tone frames (blue) and post-tone frames (red).
The tone, played at sample t=42 (marked by an arrow), is distinctively recovered from the data.
(c) First 11 eigenvectors of time embedding. Each column is an eigenvector $\psi_{t,\ell}\in\mathbb{R}^{119} \; \ell \in \{1,...,11\}$.
In general, the eigenvectors take the form of harmonic functions at different scales. 
Time $t=42$ (the tone) is apparent (marked by the box). 
Some eigenvectors correspond to harmonic functions over the entire trial (e.g., $\psi_{t,1}$), while some are localized in the pre-tone region (e.g., $\psi_{t,9}$), and some in the post-tone region (e.g., $\psi_{t,11}$).
}
\label{fig:time_embed}
\end{figure*}
Figure~\ref{fig:time_embed} presents the 3D embedding of the time frames, where each 3D point is colored by the time index $t \in \{1,...,119\}$ (a).
The embedding clearly organizes the time frames through the various repetitive experiments into two dominant clusters: ``pre-tone'' and ``post-tone'' frames (Fig.~\ref{fig:time_embed}(b)), where the tone signifies the cue for the animal to begin the hand reach movement. 
We emphasize that this prior information was not used in the data-driven analysis. 
The embedding in effect isolates the time where the auditory tone is activated for the subject to reach for food.

Figure~\ref{fig:time_embed}(c) presents the first eleven non-trivial eigenvectors $\{\psi_d(t)\}_{d=1}^{11}$ obtained by the decomposition of the affinity matrix on the time dimension.
Some eigenvectors correspond to harmonic functions over the entire interval $t \in [1,119]$.
However, some are localized either on the pre-tone region (e.g., $\psi_{t,9}$) or on the post-tone region (e.g., $\psi_{t,8}$ and $\psi_{t,11}$).
In addition, each eigenvector captures the time at varying scales.
This result demonstrates the power of our analysis; it shows that in a completely data-driven manner, a Fourier-like (harmonic) basis is attained. 
However, in contrast to the ``generic" Fourier-basis, which is fixed, the obtained basis is data adaptive and captures and characterized true hidden phenomena related to external stimuli (the tone) and to different patterns of behavior (before and after the tone). 

Thus, the embedding provides a verification of the knowledge we have regarding the time dimension in terms of regions of interest, and enables to pinpoint specific times of interest, essentially capturing the ``script'' of the trial.
We do not present the local decomposition of the time frames via the flexible tree since it is not of interest, as this dimension is smooth, and therefore is just decomposed into local temporal neighborhoods.

We next examine the analysis of the trial dimension, i.e., organization of the global time.
\begin{figure}[th]
\centering{\includegraphics[width=0.95\linewidth]{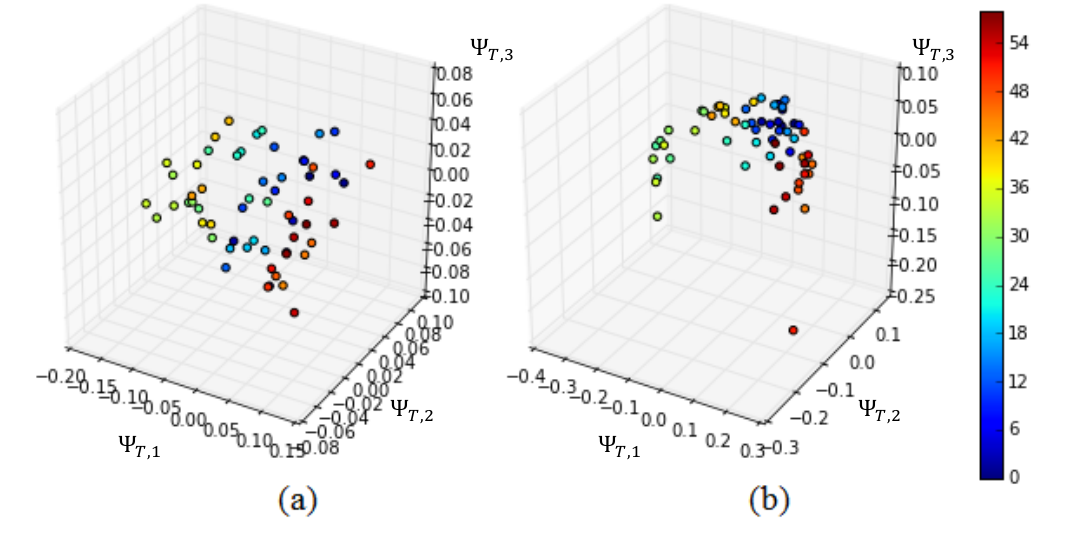}}
\caption{The 3D embedding of the 2D trial slices (Fig.~\ref{fig:data3d_slices}(left)) of all the trials $T\in\{1,...,59\}$. 
Each trial slice is represented by a single 3D point, colored by the trial index (here as well, the trial index was not taken into account in the analysis).
(a) Initial trial embedding based on cosine affinity. (b) Trial embedding derived from bi-tree multiscale metric. Trials are clustered in three main groups, where red and blue clusters are closer together.
}
\label{fig:trial_embed}
\end{figure}
In Fig.~\ref{fig:trial_embed}, we compare the embedding of the trials, obtained from the initial cosine affinity (a), and from the bi-tree multiscale metric (b). 
The points are colored by the trial index where blue corresponds to control trials (1-19), green-orange trials corresponds to the first silencing trials (19-44), and red corresponds to last silencing trials (45-59).
Our tri-geometry analysis yields an embedding (Fig.~\ref{fig:trial_embed}(b)) in which the blue and red points, corresponding to the first and last trials, respectively, are grouped together. 
This clearly indicates the temporal effect of silencing the somatosensory cortex on the activity of motor cortex. 
This is a promising result since solely from the neuronal activity, the data is self-organized functionally according to the brain activity manipulation we performed without the need to provide this information during the analysis. 
This result leads us to hypothesize that our silencing manipulation has a lag, and also that it expires over the duration of the experiment. 
Our analysis recovers hidden biological cues and enables accurate indication of pathological dysfunction driven by neuronal activity evidence.

To highlight the contribution of our approach in the analysis of such data, we compare our embedding to the 3D diffusion maps obtained by the initial cosine affinity (Fig.~\ref{fig:trial_embed}(a)), which does not exhibit any particular organization.
Thus, the refinement via iterative application of the algorithm is essential.
The multiscale local organization via the trees and coupling of the dimensions via the metric contribute to deriving a meaningful global embedding.

\begin{figure}[t]
\centering{\includegraphics[width=0.80\linewidth]{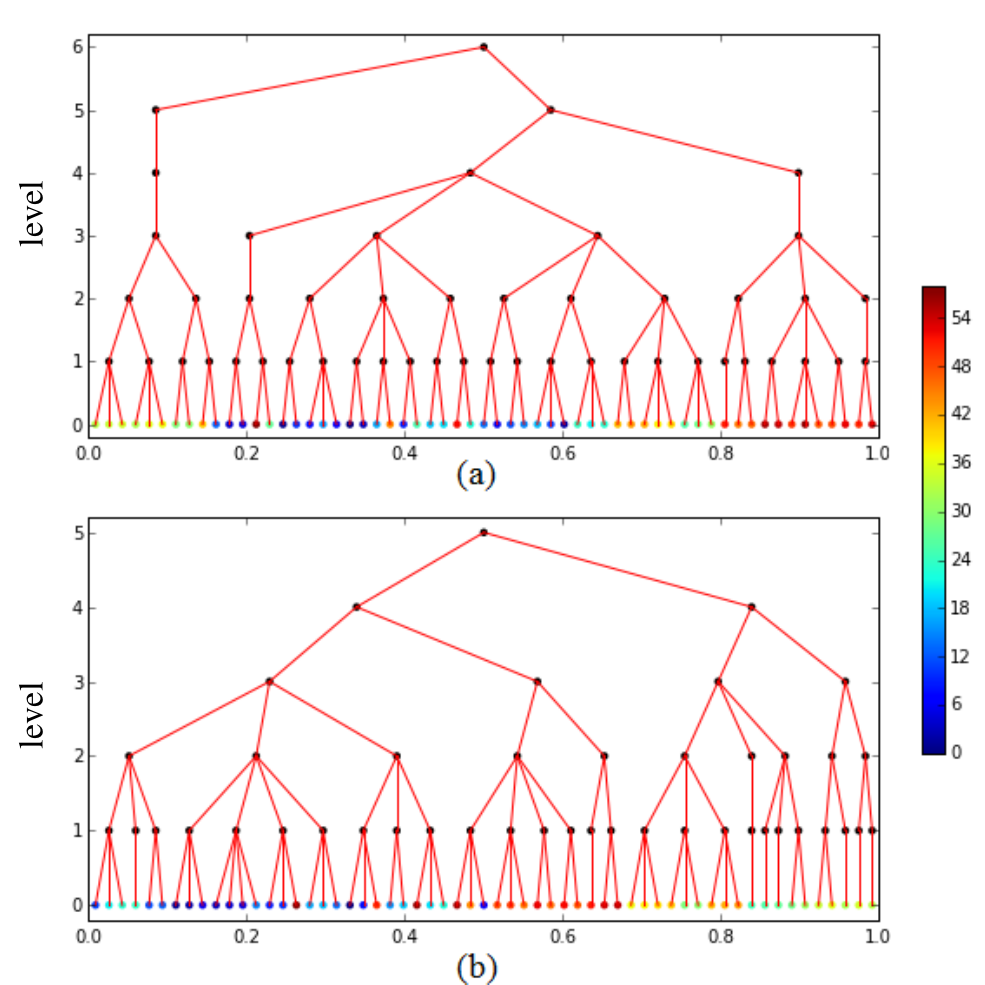}}
\caption{Flexible tree of trials ($T\in\{1,...,59\}$). The leaves are colored by trial index. 
(a) Tree corresponding to initial trial embedding in Fig.~\ref{fig:trial_embed}(a). (b) Tree corresponding to bi-tree multiscale metric embedding in Fig.~\ref{fig:trial_embed}(b). This tree better captures the nature of the trials, separating the pathological dysfunction caused by the silencing from the normal trials. }
\label{fig:trial_tree}
\end{figure}
\begin{figure*}[t!]
\centering{\includegraphics[width=0.9\linewidth]{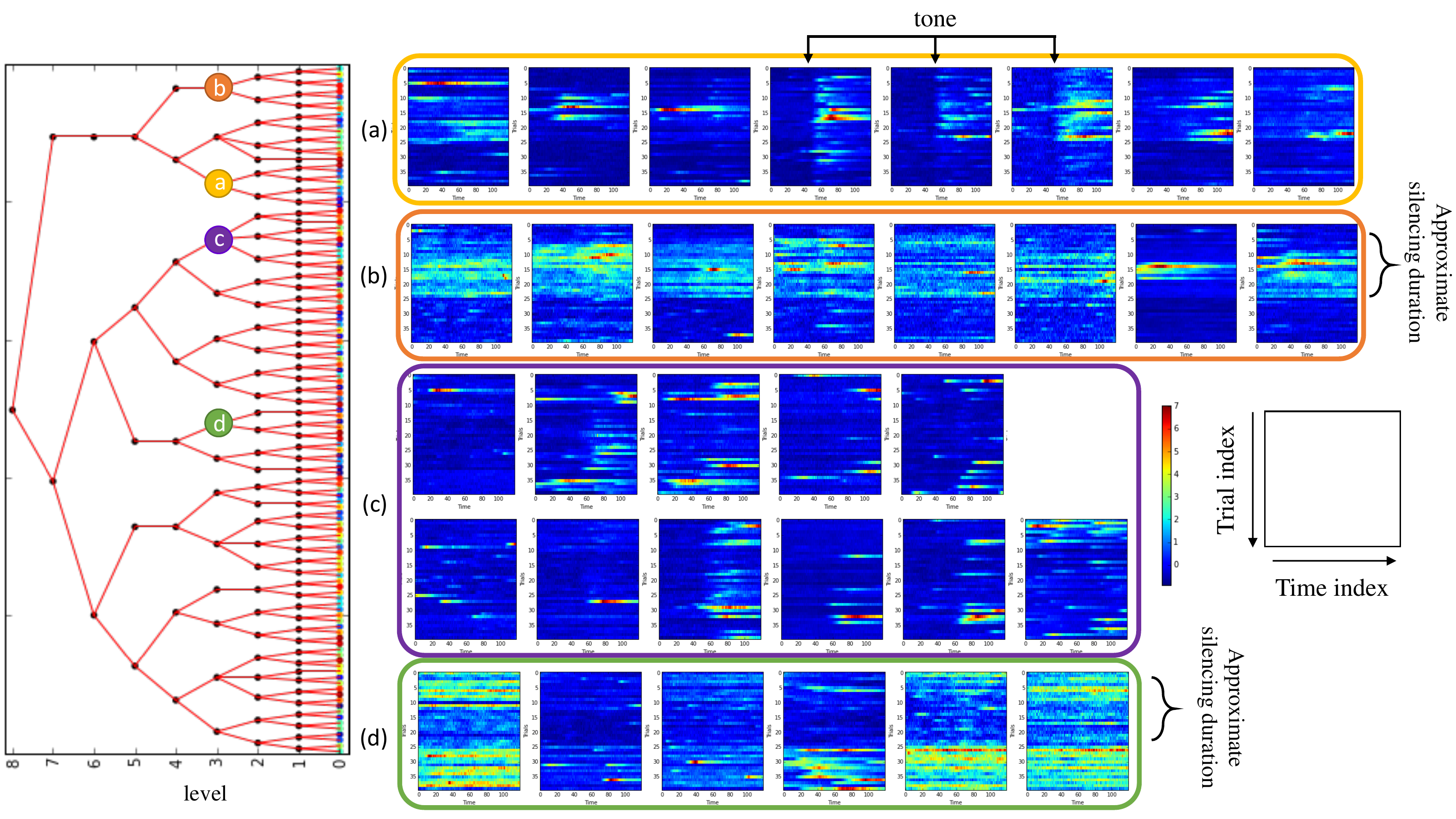}}
\caption{Neuron tree for the silencing trials for iteration $n=2$.
To demonstrate the organization obtained by the tree, we highlight several interesting tree folders from level $l=3$, marked with different colors and letters. 
Neurons belonging to the highlighted folders are grouped together, with a colored border corresponding to the folder color.
Each neuron has been reorganized as a 2D matrix of $n_T \times n_t$ ($40 \times 119$).
The neurons are grouped together according to similar properties. 
(a) Yellow folder: 8 neurons that are active only at or after the tone (vertical separation), and mostly in trials under the effect of the silencing (horizontal separation).
First three are associated with the tone itself, 5 right are associated with post-tone activity. 
(b) Orange folder: 8 neurons that were dominant mostly in trials under the effect of the silencing (horizontal separation), but are not sensitive to tone.
(This node is joined with the yellow node at level $l=5$).
(c) Purple folder: 11 neurons that were mostly active during trials not under the effect of the silencing, 8 of which are active after the tone (vertical separation). 
(d) Green folder: 5 neurons that were silenced by the manipulation (horizontal separation). 
}
\label{fig:sensors}
\end{figure*}
The improved clustering of the trials achieved by the bi-tree multiscale metric is also apparent when examining the flexible trees obtained from the two embeddings (Fig.~\ref{fig:trial_tree}).
The leaves are colored by the trial index as in the embedding.
The tree obtained from the new embedding better separates the trials in which the pathological dysfunction caused by the silencing is evident from the normal trials.
Remembering that flexible trees are constructed bottom-up using the embedding coordinates, this validates the claim that proximity in the embedding space captures the global temporal trend in the data.

To analyze the neurons, we split the data into two parts and analyze each separately, as this enables us to discover both behavioral patterns and pathological dysfunction.
First, we examine the 40 trials composing the silencing trials.
The neurons were preprocessed by subtracting the mean of each neuron over all trials, and normalizing it by its standard deviation across all trials.
This enables us to examine the increase and decrease of activity in the neuron without being sensitive to the intensity of the measurements.

Figure~\ref{fig:sensors} presents the multiscale hierarchical organization of the 2D slices of all the neurons in a flexible tree $\mathcal{T}_r^{(2)}$, obtained after two iterations of our analysis, highlighting several interesting tree folders from level $l=3$.
Neurons composing four folders from this level are presented.
The folders are marked in different colors on the tree and the neurons belonging to each folder are grouped together, with a border in corresponding color to the folder.
Each neuron has been reorganized as a 2D matrix of size $n_T \times n_t$ ($40 \times 119$).
The neurons are grouped together according to similar properties and the displayed folders clearly relate to pathological dysfunction.
For example, the orange folder consists of neurons that are active only during trials under the effect of somatosensory silencing (horizontal separation). 
The yellow folder consists of neurons that are active only at or after the tone (vertical separation), and mostly in trials under the effect of the silencing (horizontal separation). 
In contract, the purple folder consists of neurons, which are active after the tone but during trials without the silencing effect. 
Finally, the green folder consists of neurons, which were silenced by the manipulation. 
This leads us to hypothesize, as with the trial analysis, that the effect of the somatosensory silencing has a slight delay, and in addition that it wears off after a certain number of trials, since the experiment was very long.
Our analysis groups  neurons demonstrating the same activity patterns together in an automatic data-driven manner without manual intervention. 

The silencing trials enable us to analyze the neurons in terms of how they are affected by the introduced virus.
We now treat the 19 control trials, which allows us to analyze the behavioral aspect of the neurons without external intervention.
In Fig.~\ref{fig:sensors_control}, we display the neuron tree, $\mathcal{T}_r^{(1)}$, obtained after one iteration of our analysis, and examine folders for levels $l=2,3,4$.
Neurons composing five folders from this level are presented.
The folders are marked in different colors on the tree and the neurons belonging to each folder are grouped together, with a border in corresponding color to the folder.
Each neuron has been reorganized as a 2D matrix of size $n_T \times n_t$ ($19 \times 119$).
We use the labeled ``at mouth'' event and the prior information on the time of the auditory tone to analyze the results.
The binary labels indicating ``at mouth'' activity has also been reordered as a 2D matrix of size $n_T \times n_t$ ($19 \times 119$), and is displayed in within the black border.

The results indicate that neurons are grouped by similarity, clearly related to the behavioral data.
The upper two folders (red and orange) show increased activity before and during the auditory tone. 
The next three folders show increased activity after the tone.
The yellow folder composes of neurons that are activity during different trials, regardless of ``at mouth'' activity.
The purple folder, on the other hand, contains neurons that are active post-tone, almost entirely during which were successful, i.e., the subject managed to eat the food pellet, indicated by continuous ``at mouth'' labeling till the end of the trial.
Finally, the green folder is composed of two neurons with the opposite activity.
They are most active post-tone during trials in which the subject failed to eat the food pellet.
Note that this analysis is data-driven, i.e., no prior information on the event labels is used in grouping the neurons.
\begin{figure*}[th]
\centering{\includegraphics[width=0.9\linewidth]{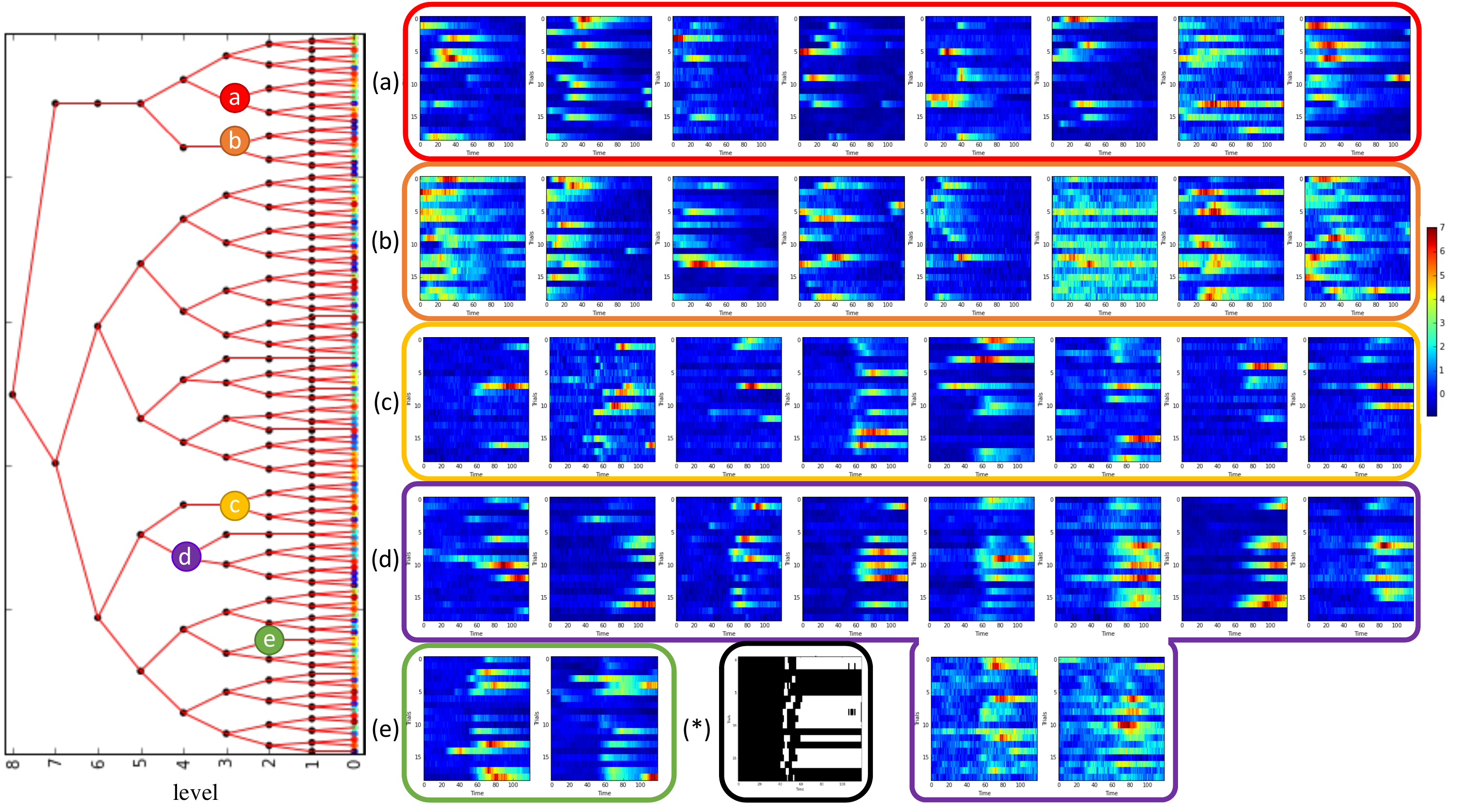}}
\caption{Neuron tree for control trials for iteration $n=1$. We highlight several interesting tree folders from level $l=2-4$, marked with different colors and letters.
Each neuron has been reorganized as a 2D matrix of $n_T \times n_t$ ($19 \times 119$).
Neurons belonging to the highlighted folders are grouped together, with a colored border corresponding to the folder color.
(a-b) Red folder (8 neurons) and orange folder (8 neurons) in level $l=3$ are active before the tone. (Note these nodes are joined at level $l=5$).
(c) Yellow folder: 8 neurons that are active post-tone.
(d) Purple folder ($l=4$): 10 neurons that are active post-tone only during trials which were labeled as ``at mouth". 
(e) Green folder: 2 neurons that with significant activity post-tone only during trials which were labeled as ``miss". 
(*) Black border contain binary labeling of at mouth event, ordered as $T \times t$ matrix.
}
\label{fig:sensors_control}
\end{figure*}

In analysis of the neurons, the main contribution is the produced partition tree.
The global embedding did not yield meaningful results, and the examination of local folders in the tree was most informative.
Note that we are looking at a limited set of ``sensors'' since the neurons were manually grouped together into ROIs.
In future work we intend to analyze a larger group of sensors, by examining all pixels acquired from the 2-photon imaging video separately.
We know from previous work that increasing the number of sensors is typically beneficial to the iterative analysis.
This will remove any introduced biases yielded by the pre-processing and enable to identify spatial structures not limited to ellipses.

Our experimental results demonstrate that our approach identifies for the first time (to the best of our knowledge), solely from observations and in a purely data-driven manner:
(i) functional subsets of neurons, (ii) activity patterns associated with particular behaviors, and (iii) the dynamics and variability in these subsets and patterns as they are related to context and to different time scales (from the variability within a trial, to a global trend in trials, induced by the silencing method. 
In analyzing the intra-trial time dimension, we pinpoint the time of the auditory trigger, and separate the time frames into multiscale local regions, before and after the trigger.
Finally, in organizing the trials, we are able to both separate the trials to ``success'' and ``failure'' cases, and to determine a global trend that relates to an introduced external intervention.
Thus, these methods lay a solid foundation for modeling the sensory-motor system by providing sufficiently fine structures and accurate view of the data to test our hypotheses, within an integrated computational theory of sensory-motor perception and action.

We note that conventional manifold learning tools did not yield any intelligent data organization for this case.
Thus, organizing the neurons or the time samples separately by a 1D geometry using conventional manifold learning methods is inappropriate for this complex data.
The fact, demonstrated here, that the neuronal activity of different types of neurons is correlated only during specific times, and might be random otherwise, verifies the need for coupled organization analysis which simultaneously organizes time, trials and neurons into tri-geometries.

\section{Conclusions}
In this paper we presented a new data-driven methodology for the analysis of trial-based data, specifically trials of neuronal measurements.
Our approach relies on an iterative local to global refinement procedure, which organizes the data in coupled hierarchical structures and yields a global embedding in each dimension.
Our analysis enabled extracting hidden biological cues and accurate indication of pathological dysfunction extracted solely from the measurements.
We identified neuronal activity patterns and variability in these patterns related to external triggers and behavioral events, at different time scales, from recovering the local ``script'' of the trial, to a global trend across trials.
In this paper we focused on neuronal measurements, but our approach is general and can be applied to other types of trial-based experimental data, and even to general high-dimensional datasets such as video, temporal hyperspectral measurements, and more.

In future work, we intend to analyze the neuronal measurements from the two-photon imaging without clustering them into ROIs. 
This significantly increases the number of ``sensors'' and should enable to learn complex spatial structures in the cortex.
Furthermore, our analysis can be extended to higher dimensions, e.g., incorporating behavioral data as a fourth dimension in the neuronal measurements.

\bibliographystyle{IEEEtran}
\bibliography{mybib}

\vfill\pagebreak
\end{document}